\documentclass[letterpaper, 10 pt, conference]{ieeeconf}  

\IEEEoverridecommandlockouts                              
\usepackage{amsfonts}
\usepackage{graphicx}
\usepackage{amsmath} 
\usepackage{algorithm}
\usepackage[noend]{algpseudocode}
\usepackage{tikz}
\usepackage{color}
\graphicspath{ {./images/} }
\usepackage{array}
\usepackage{subcaption}
\usepackage{mathtools}

\usetikzlibrary{arrows.meta}

\DeclarePairedDelimiter{\ceil}{\lceil}{\rceil}
\newtheorem{theorem}{Theorem}
\newtheorem{lemma}{Lemma}
\newtheorem{invariant}{Invariant}

\newtheorem{definition}{Definition}[section]

\usepackage[normalem]{ulem}
\newcommand{\fast}[1]{{\color{blue}#1}}
\newcommand{\gen}[1]{{\color{blue}#1}}

\title{\LARGE \bf
A Hierarchical Model for Fast Distributed Consensus \\ in Dynamic Networks
}

\author{Timothy Castiglia, Colin Goldberg, and Stacy Patterson
\thanks{T. Castiglia, C. Goldberg, and S. Patterson are with the Department 
        of Computer Science, Rensselaer Polytechnic Institute, 
        110 8th St, Troy, NY 12180,
        {\tt\small castit@rpi.edu, goldbc@rpi.edu, sep@cs.rpi.edu}.}
    \thanks{This work was supported in part by NSF grants CNS-1553340 and CNS-1816307.}
}

\newcommand{\sep}[1]{}
\newcommand{\tim}[1]{}

\begin{document}

\maketitle 
\thispagestyle{empty}
\pagestyle{empty}

\begin{abstract}
We present two new consensus algorithms for dynamic networks. The first, Fast Raft, is a variation on the Raft consensus algorithm that reduces the number of message rounds in typical operation. Fast Raft is ideal for fast-paced distributed systems where membership changes over time and where sites must reach consensus quickly. The second, C-Raft, is targeted for distributed systems where sites are grouped into clusters, with fast communication within clusters and slower communication between clusters. C-Raft uses Fast Raft as a building block and defines a hierarchical model of consensus to improve upon throughput in globally distributed systems. We prove the safety and liveness properties of each algorithm. Finally, we present an experimental evaluation of both algorithms in AWS.
\end{abstract}

\section{Introduction}

State machine replication is a foundational tool that is used to provide availability and fault tolerance in distributed systems.
Sites use consensus algorithms as a method to achieve agreement on the order of updates.
Such algorithms appear in many of today's distributed systems.
Chubby~\cite{chubby} utilizes consensus for the
Google File System to elect a master server that coordinates
replication of files.
Autopilot~\cite{isard2007autopilot} creates fault-tolerant
replicas in Microsoft's data centers around the world using consensus.
Quorum~\cite{buterin2013ethereum}, a permissioned 
blockchain technology based on Ethereum,
achieves data consistency through consensus.

Modern distributed systems are often 
large-scale, globally distributed, and dynamic. 
Data centers used by Google, Amazon, and Microsoft are 
spread globally to reach people across the world~\cite{rad2009survey}. 
Cellular networks
face latency across continents and are highly dynamic with mobile
participants~\cite{tzanakaki2013virtualization}. 
IoT/edge networks often face failures and unreliable messaging,
requiring fault-tolerant systems~\cite{mqttrss}.


Two of the most widely adopted consensus algorithms, Paxos~\cite{lamport2001paxos}
and Raft~\cite{ongaro2014consensus}, are not designed with 
such systems in mind.
These algorithms specify methods to maintain a replicated
log in an asynchronous system with message loss and crash failures.
Both require several rounds of messaging between
a leader and a majority of sites. Such an operation can
lead to high latency for a client's value to be published.
This can be problematic in modern systems that service millions
of clients every day. Fast Paxos~\cite{lamport2006fast} is a variation
on Paxos that reduces the number of message rounds in typical operation
to mitigate this problem.

Raft specifies methods for dealing with membership
changes, but with the assumption that sites are added and removed by a system 
administrator, such that membership changes are never unexpected by the cluster.
Vertical Paxos~\cite{lamport2009vertical}, another Paxos variation, 
similarly depends on an auxiliary configuration master for reconfiguration.
However, in many distributed systems, membership changes
may be sudden, and may occur silently.
Many variations on Paxos variations have been proposed that 
improve consensus for modern systems, such as Dynamic Paxos~\cite{nawab2018dpaxos},
EPaxos~\cite{moraru2013there}, and Delegator Paxos~\cite{liu2016d}. These algorithms either do
not handle dynamic membership explicitly, or require the same number
message rounds as Paxos and Raft do.

To address the limitations of previous works, 
we propose Fast Raft, 
with the goal of improving consensus speed in dynamic and
globally distributed systems. 
Fast Raft is based on Fast Paxos, 
and speeds up typical consensus operation by
paying a slightly higher penalty for recovery in the face of failures.
Fast Raft specifies details omitted from Fast Paxos,
such as leader election and unexpected membership changes.

While Fast Raft addresses some of the limitations of consensus algorithms 
in fast-paced dynamic networks, it still requires communication between 
a leader and all sites in the system to reach consensus.
This communication paradigm may not scale to global networks, where such 
all-to-one communication is both time and bandwidth consuming.
In such  systems, a hierarchical model can be beneficial.
Here, sites are grouped into clusters. The bulk of the computation is performed 
within each cluster, with results combined globally
at a lower frequency. Such a hierarchy can be either physical
or logical.
Physical hierarchical models have been utilized in database 
systems~\cite{connolly2005practical}, sensor networks~\cite{bandyopadhyay2003energy} 
and edge computing~\cite{skala2015scalable}.
Federated learning~\cite{shokri2015privacy} is an example
of a logical hierarchy used to achieve both privacy and improved
model training.
Such a logical hierarchy can be applied to blockchain as well.
Desu et al.~\cite{dasu2018unchain} proposed a partitioning of blockchain into
a hierarchy of sub-chains to improve scalability
and reduce energy consumption. 

To complement Fast Raft, we propose a hierarchical consensus algorithm that we call \emph{Clustered Raft} or \emph{C-Raft}.
C-Raft is
designed to improve throughput in globally distributed systems by
utilizing a hierarchical structure. Rather than all sites taking
part in consensus, sites in a cluster run local consensus,
then cluster leaders replicate entries in batches to other
clusters. C-Raft is defined with Fast Raft as a building block, 
further improving the speed of consensus in typical operation.
We provide a specification for C-Raft and prove that it
satisfies safety and liveness properties.
We also provide experimental results of the performance of
Fast Raft and C-Raft against classic Raft in AWS. On average, Fast
Raft, is twice as fast as classic Raft if message loss is
below $5\%$
while C-Raft achieves $5$x the throughput of Raft 
in a globally distributed system.


The rest of the paper is organized as follows. Section~\ref{formulation.sec}
formulates the problem and states the properties we wish to satisfy.
In Section~\ref{background.sec}, we provide an overview of the
classic Raft and Fast Raft algorithms.
Section~\ref{fast.sec} presents detailed pseudocode for
Fast Raft, as well as proves the properties it satisfies.
We present the C-Raft algorithm and prove its safety and liveness in Section~\ref{craft.sec}.
Section~\ref{exp.sec} presents
experiments comparing Fast Raft and C-Raft
with classic Raft. 
We discuss related works in
Section~\ref{related.sec}, and
conclude in Section~\ref{conclusion.sec}.

\section{System Model}
\label{formulation.sec}

We consider a system that consists of a set of sites. 
Sites can crash and may recover, and we assume each site has a means 
of stable storage that can be read from upon recovery.
Messaging between sites is asynchronous with potential message loss.  
Sites may join or leave the system over time. 
It is assumed joining sites have a means of contacting 
sites already in the system.

A \emph{global log} is replicated at every site. The log
is an infinite array of log entries, indexed using the natural numbers. 
Sites \emph{propose} entries to
be inserted in the global log. Sites reach
consensus on the entry, at which point the entry is \emph{committed}.

Our goal is to satisfy the following properties.

\begin{definition}{\textit{Safety}} \label{safety.def}
If a site commits an entry at some index in the log,
no site can commit a different entry at the same index.
\end{definition}

\begin{definition}{\textit{Liveness}} \label{liveness.def}
For any proposed
entry $v$, it will eventually be the case that $v$
is stored in the log at all sites.  
\end{definition}

By satisfying safety, every site will have the same entries
in each index of the logs, thus defining a \emph{total order} 
for every entry committed. 
By the findings of Fischer, Lynch, and Paterson~\cite{fischer1985impossibility}, 
we cannot guarantee liveness in a fully asynchronous system with crash failures
without additional conditions. Along with the specification of Fast Raft and C-Raft, we
prove they satisfy safety, and identify conditions for liveness. 

\section{Background and Algorithm Overview} \label{background.sec}
In this section, we provide an overview of classic
Raft, and the modifications made to construct Fast Raft.

\subsection{Classic Raft}

In classic Raft, time is split into \emph{terms}
numbered in a monotonically increasing manner.
Sites take on one or more roles in a term.
\emph{Proposers} propose new log entries to be appended to the log.
A unique \emph{leader} gathers proposals of log
entries and coordinates consensus on the entries. There are
\emph{followers} that participate in consensus on new log entries.
Finally, there are \emph{candidates} that attempt to be elected
as a new leader if they suspect the current leader of having failed. 
In a typical term, one candidate is elected
as a leader, then consensus on log entries can begin. We describe the means of 
consensus and election below.

The log can contain both proposed entries and committed entries.
The $commitIndex$ is a value stored at each site that determines
which values in the site's log are committed. Entries in indices 
at or before $commitIndex$ are considered committed by the site.
Proposers will send proposed entries to the term's leader.
The leader appends proposed entries to its log and
it sends an $AppendEntries$ message to followers
to also append the entry. When the leader receives acknowledgement
that a majority, or \emph{classic quorum}, 
has appended the entry, the leader's $commitIndex$ is updated, indicating that the value has been committed.
The leader then notifies the proposer. On subsequent $AppendEntries$,
the leader will include its $commitIndex$ and followers can update
their own $commitIndex$ accordingly.

\begin{figure}[t]
    \centering
    \includegraphics[scale=0.5]{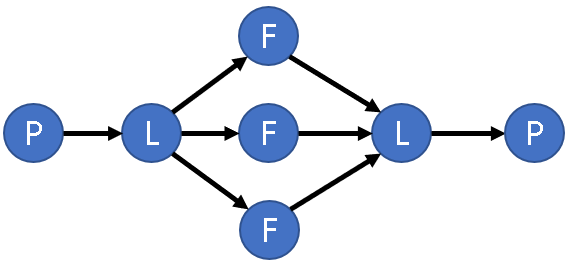}
    \caption{Message flow of appending an entry in classic Raft. 
    $P$ is the proposer, $L$ is the leader, $F$ are the followers.}
    \label{classic.fig}
\end{figure}

In our system model, it is possible for messages to be lost, or
sites to crash. If a leader fails, then a new leader
must be elected. As a means of failure detection, the leader 
maintains a heartbeat with followers. If a follower does not
receive a heartbeat message within some \emph{election timeout},
the follower converts to a \emph{candidate}, increments its term
number, and begins an election. The candidate sends $RequestVote$
messages to all sites, who are notified of the term change.
Messages from leaders of lower terms will be rejected and notified
that a new term has begun. 
Upon receiving a $RequestVote$ message, a site sends a vote
to the candidate if the candidate's log is at least as \emph{up-to-date}
as its own. A log with the most log entries from the most recent term
is considered the most up-to-date. If a majority of sites
send votes to the candidate, it becomes the leader for that term.

It is possible that multiple followers time out and become
candidates for the same term. This may lead to no candidate 
being elected in a term. Candidates that fail to receive enough
votes wait for a randomized timeout before incrementing the
term number again and retrying an election. These randomized
timeouts ensure that a leader is eventually elected with high
probability. A Candidate
converts back to a follower if it receives an $AppendEntries$
message from a newly elected leader at the same or higher term 
number. 

To allow for dynamic membership, Raft defines a special
type of log entry called a \emph{configuration} entry. This entry
contains a list of all \emph{voting members} of the system,
the sites can take part in consensus. If a site receives a consensus-related message from a site that is not in the configuration, the
message is ignored. Each site considers the last
appended configuration entry to be its current configuration.
Raft assumes a system administrator proposes configuration 
changes to the leader, and that only one site is added
or removed from the configuration at a time. Without this
restriction, it may be possible that messages are lost in
consensus on reconfiguration, and there is a group of sites with the
old configuration and quorum size.
If the old quorum does not overlap with every 
new quorum, two leaders may be elected, violating safety.
 
When a site is proposed to be added to a configuration, the leader
first catches up the site on all entries in the leader's log. 
During this period, the joining site is considered
a \emph{non-voting member} of the system, and it cannot take
part in consensus yet.
Once caught up, the new configuration is appended 
to the leader's log and consensus is run. Once a majority has
appended the new configuration entry, the joining site
is considered a voting member. A similar protocol is followed
for a site leaving the configuration.

\subsection{Fast Raft}

When there are no concurrent proposals, Fast Raft decreases
the number of message rounds from three to two before an entry is committed.
To achieve this without violating safety, Fast Raft requires
a recovery mechanism in the face of failures.

Fast Raft has the same roles as classic Raft, but consensus works
differently. Fast Raft provides two methods of committing values
to the log: a \emph{fast track} and a \emph{classic track}.
The fast track, as the name indicates, requires fewer rounds to
commit a value. The classic track is followed if the fast track
fails due to message loss or concurrent proposals. We describe
these tracks below.

When a proposer has a new entry $e$ to be inserted in the log at an
index $i$, 
rather than sending the entry to the leader, the proposer
sends the entry to all sites. Upon receiving the new entry, 
a site $a$ inserts the proposed entry to its log at index $i$.
This contrasts with classic Raft, where the log is treated as a 
growing list that is always appended to. In the case of Fast Raft,
site $a$ may miss a proposal for an entry at index $j < i$. 
Follower $a$ will insert the entry at $i$, leaving index $j$ empty.
After inserting, the follower forwards the message to
the leader, indicating that $a$ has \emph{voted} for entry $e$.
At this point, entry $e$ in site $a$'s log is marked as \emph{self-approved}, 
as site $a$ inserted the entry itself. 

\begin{figure}[t]
    \centering
    \includegraphics[scale=0.5]{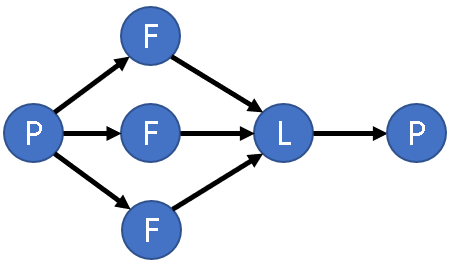}
    \caption{Message flow of inserting a new entry in Fast Raft.
    $P$ is the proposer, $L$ is the leader, $F$ are the followers.}
    \label{fast.fig}
\end{figure}

The leader gathers these votes for entries at index $i$.
After receiving votes, the leader makes a decision on the entry that should
be inserted into its own log. 
Let $M$ be the number of sites in the configuration.
If $\ceil{\frac{3M}{4}}$ (or more) sites, a \emph{fast quorum},
have voted for the same entry $e$, then the leader commits $e$.
The message
flow of the fast track in Fast Raft is shown in Figure~\ref{fast.fig}. 

With this method, the leader may receive 
votes for different entries, and may be missing votes if messages are lost.
The leader will wait for at least a classic quorum of votes. 
There are only two scenarios: either a fast quorum has appended the same entry
or not. In the first scenario, the leader needs to insert this entry.
Otherwise, the leader can insert any entry. Due to the chosen quorum sizes,
if a fast quorum has inserted an entry, it will be the entry with the most
votes in any classic quorum. We provide intuition as to why this is in
the following example.
Consider a scenario with five sites. 
For the same index, proposals $e$ and $f$ were proposed. Four sites, a fast quorum, 
insert $e$, and one site inserts $f$. Vote messages are lost and the 
leader only receives from three sites, a classic quorum. No matter which
three sites the leader receives votes from, two of the sites must have voted
for $e$. Thus, if a entry has been inserted by a fast quorum, 
the leader will always insert that entry.

If the leader inserts an entry but does not commit it, we 
revert to the classic track, which is identical to classic Raft.
The leader sends $AppendEntries$ messages to have a classic quorum
insert the entry that had the most votes.
Entries inserted on the classic track are marked as \emph{leader-approved}.
If a follower receives an entry from the leader that it already inserted,
it will update it to be leader-approved.
Note, the classic track only occurs after attempting the fast
track, and thus, in this situation, we suffer the penalty of an 
extra message round compared to classic Raft.

Leader election follows the same flow as in classic Raft with some
alterations. Since
proposers can send directly to followers, the definition of up-to-date
is modified to only include leader-approved entries.
Self-approved entries cannot be considered in this check, as proposers
can send an arbitrarily large number of proposals to a follower that
ultimately may not have been agreed upon. 
Once the most up-to-date candidate is elected, Fast Raft runs
a recovery algorithm. Self-approved entries were not 
considered in the election, and need to be evaluated to ensure safety.
All followers resend their self-approved entries to the newly elected leader.
If a leader from a previous term committed any of these entries, then
a there will be a fast quorum that has inserted the entry, and the new
leader will make the same decision as previous leaders and commit the entry.

Fast Raft deals with dynamic membership similarly to classic Raft.
However, we do not assume a system administrator proposes configuration
changes. Instead, joining and leaving sites send join 
or leave requests to the leader. It is now the role of the leader to
ensure that only one site joins at a time by processing join requests
sequentially. Any sites with
pending joins are considered non-voting members of the system.
Unlike Raft, Fast Raft has specification for when a site leaves
the system silently, as sites may not always propose a leave request
before leaving the system.


\section{Fast Raft}
\label{fast.sec}

We now dive into the details and pseudocode for the Fast Raft algorithm
and provide a proof of safety and discussion of liveness.
Some parts of the Fast Raft specification are the same as in classic
Raft. We mark items that are new to Fast Raft in blue and with a $\dagger$ symbol.

\subsection{Site State}

Each site stores some state that is persistent in stable
storage, as well as volatile state. 

\vspace{0.5mm}
\noindent\fbox{%
\parbox{\linewidth}{%
\textbf{Persistent state on all sites:}
\begin{itemize}
    \setlength\itemsep{0em}
    \item $currentTerm:$ latest term sent to site (initialized to 0)
    \item $votedFor:$ candidateId that site voted for in current term (initialized to null)
    \item $log[]:$ global log to be replicated 
    \item $lastLogIndex:$ last index appended to log 
    \item \fast{$lastLeaderIndex:$ last index appended to log by the leader.$^\dagger$}
    \item $configuration:$ last configuration appended to the log.
    \item $commitIndex:$ index of highest log entry known to be
        committed (initialized to 0)
\end{itemize}}}
\vspace{0.01mm}

The variables $currentTerm$ and $votedFor$ are 
used for leader election to help determine which site will
be the new leader. This is discussed more in Section~\ref{election.sec}.
The $lastLogIndex$, $lastLeaderIndex$, and $configuration$ can be determined based
on the log, but are separate variables here for notation
convenience.

As discussed in the overview, the $commitIndex$ indicates 
the entries in the log that are committed. Any entry in an
index larger than $commitIndex$ is not yet committed.
As $commitIndex$ is in volatile state, 
if a site crashes and recovers, it will need to relearn which
log entries are committed from the current leader.

\vspace{0.5mm}
\noindent\fbox{%
\parbox{\linewidth}{%
\textbf{Contents of a log entry:}
\begin{itemize}
    \setlength\itemsep{0em}
    \item $data:$ data to be replicated.
    \item $term:$ term number when the entry was appended.
    \item \fast{$insertedBy:$ either self or leader.$^\dagger$}
\end{itemize}}}
\vspace{0.01mm}

Each log entry contains some data to be replicated and
the term number in which it was added. This information is necessary for 
deciding which site is most up-to-date and should be elected leader.
The only change from classic Raft here is the $insertedBy$ value, indicating
if the entry was self-approved or leader-approved. 

%

\vspace{0.5mm}
\noindent\fbox{%
\parbox{\linewidth}{%
\textbf{Volatile state on the leader:}
\begin{itemize}
    \setlength\itemsep{0em}
    \item $nextIndex[]:$ for each site, index of the next log entry
        to send to that site (initialized to leader's
        last \fast{committed$^\dagger$} log entry $+ 1$).

    \item $matchIndex[]:$ for each site, index of highest log entry
        known to be replicated on site (initialized to 0).
    \item \fast{$fastMatchIndex[]:$ for each site, index of highest log entry
        known to be sent to the leader that matches the leader's choice 
        (initialized to 0)$^\dagger$.} 
\end{itemize}}}
\vspace{0.01mm}
\noindent\fbox{%
\parbox{\linewidth}{%
\begin{itemize}
    \item \fast{$possibleEntries[]:$ array of sets of possible entries 
            for the index in the leader's log. Each index corresponds
            to a log index. Each index holds a set of pairs, each consisting of a proposed
            entry and number of votes for that entry$^\dagger$.}
\end{itemize}}}
\vspace{0.01mm}

The leader keeps track of $nextIndex[i]$ to determine
what log entries need to be sent to site $i$. Rather than the leader
sending all entries, it only sends entries starting from $nextIndex$.
$matchIndex[i]$ indicates the latest log index that site $i$ matches with 
the leader. Once a classic quorum is at the same or higher $matchIndex$,  
the leader commits the entry at that index. 

Fast Raft has an additional array, $fastMatchIndex$, separate from $matchIndex$
for determining if an entry can be committed on the fast track or 
classic track.
Each of these arrays contain an index for a site in the configuration,
and thus are subject to size changes. When a site joins, a new
slot in each array must be added. For leaving sites, it may be beneficial for
the values associated with them to be stored in case they 
return to the system at a later time.

Unique to Fast Raft, the $possibleEntries$ structure plays an important role.
A leader in classic Raft would immediately append entries proposed to it.
However, Fast Raft needs a method by which to keep track of the votes of 
followers for a log index. The leader makes its decision on what entry to
insert or commit based on the contents of $possibleEntries$.
This decision is explained in the next subsection.

\subsection{Inserting Entries} \label{insert.sec}
Log entries are proposed to be inserted into the replicated
log at specific indices. The leader gather votes for 
log entries at an index $i$ and determines which
entries to commit. In this
section, we describe the flow of entries from proposed to
committed on both the fast track and the classic track.

\vspace{0.5mm}
\noindent\fbox{%
\parbox{\linewidth}{%
\textbf{To propose an entry:}
\begin{enumerate}
    \setlength\itemsep{0em}
    \item Send log entry to \sout{$leaderId$} \fast{all members in configuration$^\dagger$}.
    \item \fast{If log entry not committed after proposal timeout,
        resend log entry$^\dagger$.}
\end{enumerate}}}
\vspace{0.01mm}

In Fast Raft, proposers send directly to all members, 
rather than going through the leader.
In the case of concurrent proposals, a proposer's entry may be
overwritten by another proposed entry that garnered more votes. 
As such, a proposer has a \emph{proposal timeout}, a period of
time to wait until reproposing if its entry is not committed.

\vspace{0.5mm}
\noindent\fbox{%
\parbox{\linewidth}{%
\fast{
\textbf{When follower receives a proposed entry $e$ for index $i$$^\dagger$:}
\begin{enumerate}
    \setlength\itemsep{0em}
    \item If entry is duplicate and committed, notify proposer.
    \item If there is no entry at index $i$, 
        insert entry to log at index $i$.
    \item Set $e.insertedBy = self$
    \item Send $log[i]$ and $commitIndex$ to $leaderId$.
\end{enumerate}}}}
\vspace{0.01mm}

Sites insert entries sent to them and then send their vote to the leader.
If an entry already exists at index $i$, the follower does not
overwrite it.
This contrasts with Raft, in that sites do not always necessarily append
entries to their end of their log, meaning log entries of higher indices
can be inserted before lower indices if messages are lost. 
Sites need to be aware of duplicate entries, as a proposer may resend its entry if it is not committed.
The leader is treated as a follower in this scenario, and follows
the same process. 

\vspace{0.5mm}
\noindent\fbox{%
\parbox{\linewidth}{%
\fast{
\textbf{When leader receives an entry $e$ for index $k$ from site $i$$^\dagger$:}
\begin{enumerate}
    \setlength\itemsep{0em}
    \item If $e \in possibleEntries[k]$, increment count for $e$.
        Otherwise, add $e$ to $possibleEntries[k]$ with \\ $count=1$.
    \item Set $nextIndex[i] = sentCommitIndex$.
\end{enumerate}
}}}

The leader receives votes from followers for entries at an 
index $k$, and tracks the votes in $possibleEntries[k]$.
The leader also sets the $nextIndex$ for agent $i$ to the
last committed index. In Fast Raft, this is
necessary for ensuring sites stay consistent with
a newly elected leader.
We discuss this in more detail in Section~\ref{election.sec}.

\vspace{0.5mm}
\noindent\fbox{%
\parbox{\linewidth}{%
\fast{
\textbf{Periodically run by the leader$^\dagger$:}
\begin{enumerate}
    \setlength\itemsep{0em}
    \item While there exists a $k = commitIndex+1$ for which 
        at least a classic quorum of votes has been received:
    \begin{enumerate}
        \setlength\itemsep{0em}
        \item Insert entry $e$ from $possibleEntries[k]$ with highest number 
            of votes. Break ties arbitrarily.
        \item Set $e.insertedBy = leader$
        \item Update $fastMatchIndex[i]$ for all $i$ that voted for entry $e$.
        \item If $e$ is elsewhere in $possibleEntries$, set to a $null$ 
            vote to avoid inserting a duplicate entry.
        \item If there is a fast quorum that exists such that $fastMatchIndex[i] \geq k$, and 
            $log[k].term = currentTerm$:
            set $commitIndex = k$
    \end{enumerate}
\end{enumerate}}}}
\vspace{0.01mm}

New to Fast Raft, the leader periodically checks if an entry can
be committed on the fast track. As discussed in the overview, if a fast quorum
has voted for an entry, it can be committed. Otherwise, if at least a classic
quorum has voted for an index, the leader chooses the entry with 
the most votes to insert into its log. The leader then switches to 
the classic track, sending this entry to the follower to insert.

The fast track can only be taken here if the last index was committed.
This restriction is necessary since $commitIndex$ indicates
that all entries at or before the index are committed. If
any entry was able to take the fast track, then $commitIndex$
could be updated prematurely, committing log entries that
should not be. \tim{WHY USE FAST RAFT?}

\vspace{0.5mm}
\noindent\fbox{%
\parbox{\linewidth}{%
\textbf{Periodically run by the leader:}
\begin{enumerate}
    \item \textbf{Create AppendEntries message for each follower containing:}
    \begin{itemize}
        \setlength\itemsep{0em}
        \item $term:$ leader's term
        \item $leaderId:$ for followers to redirect sites
        \item $entries[]:$ log entries to insert
    \end{itemize}
\item For each follower $i$, if $\fast{lastLeaderIndex}^{\dagger} \geq nextIndex[i]$,
    include all log entries in $entries[]$ starting at $nextIndex[i]$ \fast{and
    ending at $lastLeaderIndex^{\dagger}$}.
    \item Send $AppendEntries$ messages to followers.
\end{enumerate}}}
\vspace{0.01mm}

Periodically, the leader sends $AppendEntries$ to all
followers. The message contains entries 
the follower may not have inserted yet.
Even if there are no new entries
for a site, the leader sends an empty message to
followers. This heartbeat lets sites know the 
leader is still active.

\vspace{0.5mm}
\noindent\fbox{%
\parbox{\linewidth}{%
\textbf{When a follower receives AppendEntries message:}
\begin{enumerate}
    \setlength\itemsep{0em}
    \item Reset election timer.
    \item \textbf{Create response to AppendEntries message:}
    \begin{itemize}
        \item $term:$ currentTerm, for leader to update itself
        \item $success:$ true if follower contained entry matching
                $prevLogIndex$ and $prevLogTerm$
        \setlength\itemsep{0em}
    \end{itemize}
    \item If $term < currentTerm$, set $success = false$ and respond.
    \item If an existing entry conflicts with a new one 
        \fast{overwrite the existing entry$^\dagger$}.
    \item If $leaderCommit > commitIndex$ set index to
             the minimum of $leaderCommit$ and $lastLogIndex$
\end{enumerate}}}
\vspace{0.01mm}

Sites check if their term number is higher
than the one sent. If it is, this means the
sender is no longer the leader. The follower will
send back $success = false$ to the old leader. Otherwise, the site continues
to inserting new log entries.
In Fast Raft, sites overwrite any entries at the same indices as
the entries the leader sent to them. Classic Raft would
remove any entries inconsistent with the current leader here. However,
since proposers send to followers first, the leader may be unaware
of entries for some indices. Overwriting these could violate safety. 
Thus, the leader only overwrites entries it has made safe decisions about.
At this point, the entries are leader-approved. 

Note that sites receive the leader's 
$commitIndex$ with new entries. Leaders are always the first to commit an
entry, and followers only update their own $commitIndex$ after receiving from the leader.

\vspace{0.5mm}
\noindent\fbox{%
\parbox{\linewidth}{%
\textbf{When the leader receives AppendEntries message response from follower $i$:}
\begin{enumerate}
    \setlength\itemsep{0em}
    \item If $term > currentTerm$, set $currentTerm = term$
            and convert to a follower.
    \item If $success = true$, update $nextIndex[i]$
        and $matchIndex[i]$.
    \item If an index $k$ exists such that $k > commitIndex$,
        a classic quorum of $matchIndex[i] \geq k$, and 
        $log[k].term = currentTerm$, then 
        set $commitIndex = k$
\end{enumerate}}}
\vspace{0.01mm}

The leader updates $nextIndex$ and $matchIndex$ based on
the value of $success$. If $matchIndex$
indicates a majority has inserted an entry, 
the leader commits the entry.

\subsection{Leader Election} \label{election.sec}
In contrast to Paxos, leader election in classic and 
Fast Raft is worked into the algorithm. The heartbeat message from the leader acts
as a failure detector for the followers. 

\vspace{0.5mm}
\noindent\fbox{%
\parbox{\linewidth}{%
\textbf{When election timeout occurs:}
\begin{enumerate}
    \setlength\itemsep{0em}
    \item Convert to a candidate.
    \item Increment $currentTerm$.
    \item Set $votedFor = self$.
    \item \textbf{Create RequestVote message containing:}
    \begin{itemize}
        \setlength\itemsep{0em}
        \item $term:$ candidate's term
        \item $candidateId:$ candidate requesting vote
        \item $candLastLogIndex:$ index of candidate's last \fast{leader-approved$^\dagger$} log entry 
        \item $candLastLogTerm:$ term of candidate's last \fast{leader-approved$^\dagger$} log entry 
    \end{itemize}
    \item Send $RequestVote$ message to all other sites.
    \item Reset election timer.
\end{enumerate}}}
\vspace{0.01mm}

If a follower
reaches an \emph{election timeout}, it begins leader election
by increasing its term number and requests votes
by sending $RequestVote$ messages to all 
other sites. It is important that the election timeout
cannot be shorter than message delays, otherwise it is 
possible for followers to continuously start elections,
and no progress can be made.

\vspace{0.5mm}
\noindent\fbox{%
\parbox{\linewidth}{%
\textbf{When receiving a RequestVote message from a candidate:}
\begin{enumerate}
    \setlength\itemsep{0em}
    \item \textbf{Create response to RequestVote message:}
    \begin{itemize}
        \setlength\itemsep{0em}
        \item $term:$ currentTerm, for candidate to update itself
        \item $voteGranted:$ true means candidate received vote 
        \item \fast{$selfApprovedEntries[]:$ all self-approved entries 
            in log$^\dagger$}
    \end{itemize}
    \item If $term < currentTerm$ set $voteGranted = false$ and respond.
    \item If $votedFor = null$ or $votedFor = candidateId$, 
        \fast{$candLastLogIndex \geq lastLeaderIndex$
        and $candLastLogTerm \geq log[lastLeaderIndex].term$, or
        $candLastLogTerm > lastLeaderIndex.term^\dagger$},
        set $voteGranted = true$, \fast{add all self-approved entries
        in log to $selfApprovedEntries^\dagger$}
            and respond. 
    \item Otherwise, set $voteGranted = false$ and respond.
\end{enumerate}}}
\vspace{0.01mm}

Sites that receive the $RequestVote$ message immediately 
move to the new term. This means leaders from previous terms are no longer 
leaders. A site votes for a candidate if the candidate 
is as \emph{up-to-date} as it is, or more. 
This ensures that the most up-to-date live site will be elected.

In classic Raft, the most up-to-date site has 
the most log entries from the most recent term. However, 
in Fast Raft, proposers can send directly to followers to insert 
into their log.
Instead, the most up-to-date site has the 
most recent leader-approved log entry.
This alone does not guarantee safety, as a failed site may have
committed an entry on the fast-track. Once elected, the new leader needs
to run a recovery algorithm (discussed below). 
The recovery requires knowledge of self-approved
entries in a classic quorum of sites. The recovery
ensures that if it is possible that an entry was committed,
then it will be committed by the new leader. 
As such, sites that vote for the 
candidate will also include their self-approved entries in 
their $RequestVote$ response message.

\vspace{0.5mm}
\noindent\fbox{%
\parbox{\linewidth}{%
\textbf{When a candidate receives any message:}
\begin{enumerate}
    \setlength\itemsep{0em}
    \item If message is a response to $RequestVote$ and
            $voteGranted = true$, increment votes.
    \item If $AppendEntries$ received from new leader: 
            convert to follower.
    \item If votes received from majority of sites: become leader,
        \fast{copy all self-approved entries received to $possibleEntries$$^\dagger$},
        and send an initial $AppendEntries$ heartbeat.
\end{enumerate}}}
\vspace{0.01mm}

If the candidate is considered up-to-date by a majority, 
meaning it has received their votes, then it 
converts to a leader. Otherwise, it loses the election, and will
retry after a randomized timeout. Other followers may timeout
and attempt to be elected as well.
Randomized timeouts are used to help ensure a leader
is eventually elected.

In Fast Raft, when the new leader is elected, the leader
fills its $possibleEntries$ structure with the received self-approved
entries. As specified in Section
\ref{insert.sec}, for each index where a classic quorum has sent 
their votes for self-approved entries, the leader chooses the entry with the
most votes to insert or commit to its log.
Leader-approved entries are treated the same as they are
treated in classic Raft, but self-approved entries require this resending
of log entries to ensure safety. We discuss more in Section \ref{safety.sec}. 

\subsection{Membership}

Similar to how leader election
is tied into consensus using term numbers, the membership configuration
is a special log entry. When a site wishes to join or leave 
the system, sites first runs consensus on a new log entry that
defines the change in configuration. Sites follow the configuration
that was last inserted into their log to determine quorum sizes
and which sites to communicate with. Messages from sites not listed in the
configuration are ignored.

\vspace{0.5mm}
\noindent\fbox{%
\parbox{\linewidth}{%
\fast{
\textbf{When a site wants to join the system$^\dagger$:}
\begin{enumerate}
    \setlength\itemsep{0em}
    \item Send a join request to sites in the system.
    \item If join not accepted in join timeout, resend request.
\end{enumerate}}}}
\noindent\fbox{%
\parbox{\linewidth}{%
\textbf{Upon receiving a join request:}
\begin{enumerate}
    \setlength\itemsep{0em}
    \item \fast{If not the leader, redirect to the leader$^\dagger$.}
    \item \fast{If a duplicate request, ignore$^\dagger$.}
    \item Catch up joining site.
\end{enumerate}}}
\noindent\fbox{%
\parbox{\linewidth}{%
\textbf{Upon catching up site:}
\begin{enumerate}
    \item Create a log entry for the new configuration
        including the new site.    
    \item Start consensus on new log entry. 
    \item When consensus reached (the entry is committed), 
        notify the new site.
\end{enumerate}}}
\noindent\fbox{%
\parbox{\linewidth}{%
\textbf{Upon receiving a leave request:}
\begin{enumerate}
    \setlength\itemsep{0em}
    \item Create a log entry for the new configuration
            excluding the site.
    \item Start consensus on new log entry.
\end{enumerate}}}
\vspace{0.01mm}

The original Raft paper assumes there is a system administrator 
that proposes configuration changes. In Fast Raft, sites
send join or leave requests to voting members. 
Before a site joins, the leader will catch-up the site by
sending all its log entries up until this point.
During this period, the site is considered a non-voting member.
The leader sends $AppendEntries$ messages to the joining site, but 
does consider its vote in consensus of log entries.
After the joining site is caught up, the leader will create
a new configuration entry to include it, and run consensus on it.
Once the configuration entry is committed, the joining site
can take part in consensus.

As discussed previously, to satisfy safety, 
only one site may join at a time,
which is managed by the leader. For example, if a leader
receives join requests for two sites, it will process the
join requests sequentially. It will create a configuration
entry for one site, wait until the entry is committed,
then create an entry for the other site.



\vspace{0.5mm}
\noindent\fbox{%
\parbox{\linewidth}{%
\fast{
\textbf{If site leaves silently$^\dagger$:}
\begin{enumerate}
    \setlength\itemsep{0em}
    \item If the leader left the system, a new leader is elected
        to detect the failure.
    \item Leader detects silent leave if a follower does not respond
        to enough $AppendEntries$ messages, specified by the member timeout.
    \item Send leave request for the missing site and start consensus
    on new configuration.
\end{enumerate}}}}
\vspace{0.01mm}

As we consider systems where sites can enter
or leave at any time, we have to modify our membership mechanism
from the original specification in classic Raft. Multiple sites
entering the system can be buffered by the current leader,
only allowing one to join the configuration one at a time. 
However, this mechanism cannot be employed for
sites leaving the system.
It is not reasonable to expect that a site requesting to leave will stay in contact, or resend its request if lost. 
Without some centralized system
administrator to ensure sites update configurations before a site
leaves, we must specify how our system will handle a
site leaving silently.

Just as the heartbeat message is used for followers to determine if
the leader has failed, the response to the heartbeat message 
can be used by the leader
to determine if a site has left the system silently. We introduce
a tunable \emph{member timeout}, the number of missed $AppendEntries$ responses 
before the leader initiates consensus on a new configuration that excludes the 
unresponsive follower. 
The follower may still be in the system, in which
case it will need to send a join request to return to the configuration.

\subsection{Safety}\label{safety.sec}

First, we prove safety in Fast Raft ignoring membership
changes. We then show that membership changes do not affect the
proof.
To ensure the safety property is satisfied, we introduce
some invariants that will build up to satisfy it.
\begin{invariant} \label{commit1.invar}
If a follower commits an entry at some index, then the leader has 
committed the same entry in the same index.
\end{invariant}
\begin{invariant} \label{commit2.invar}
If a leader has committed an entry at some index, then no leader
in a previous term has committed a different entry at the same index.
\end{invariant}

\begin{theorem}
    If Invariants \ref{commit1.invar} and \ref{commit2.invar}
    hold, then the safety property holds.
\end{theorem}
\begin{proof}
By Invariant~\ref{commit1.invar}, followers only commit once notified
that the leader has committed.
By Invariant~\ref{commit2.invar}, leaders will always commit
the same entries as leaders from previous terms.
A leader will never overwrite a committed entry.
As every leader and follower of every term will always commit the
same entries at the same indices, safety is satisfied.
\end{proof}

It is left to prove that Fast Raft satisfies the invariants.
\begin{lemma}
    Invariant~\ref{commit1.invar} holds.
\end{lemma}
\begin{proof}
A follower only updates its $commitIndex$ by taking the minimum
of the leader's $commitIndex$ and its own $log$'s length. It 
will never set $commitIndex$ to an index greater than the 
leader's $commitIndex$.
Therefore, the follower will only commit entries if the leader
has committed them.
\end{proof}

\begin{lemma} \label{commit2.lemma}
    Invariant~\ref{commit2.invar} holds.
\end{lemma}
\begin{proof}
We define a log entry to be \emph{chosen} for an index 
if one of the following two conditions hold:
\begin{enumerate}
\item there is a fast quorum with the same self-approved
    or leader-approved entry, or
\item there is a classic quorum with the
same leader-approved entry.
\end{enumerate}
To prove Lemma~\ref{commit2.lemma}, we must prove that once an entry
is chosen, no other entry can be chosen, and that a leader
of a term will only commit chosen entries.

First, we note that no entries for an index 
can be chosen simultaneously as every classic quorum and fast quorum intersect.
A follower may overwrite its entry for an index with an entry sent from the leader. 
Suppose an entry $v$ was chosen by a fast quorum $R$.
As proven by Zhao in~\cite{zhao2015fast},
any classic quorum of votes that the leader 
could have received on the fast track
must have a majority overlap with $R$.
Entry $v$ will have the most votes in any classic quorum,
and the leader always inserts the entry with the most votes.
Thus, when a leader sends an entry to followers,
either the entry is chosen, or no entry is chosen
for the index. So, a follower never overwrites a chosen entry. 

It is left to prove that a leader always commits a chosen entry.
If the leader of a term $t$ successfully gathers votes
of a fast quorum for the same entry, it commits the chosen entry.
Otherwise, the leader inserts an entry for the classic track. Either
this is a chosen entry, or no entry was previously chosen. The proof of safety for 
the classic track follows directly from the proof for classic Raft~\cite{ongaro2014consensus}.

Suppose the leader of term $t$ fails before committing. 
Leaders of terms $s > t$, upon election, will first 
run Fast Raft's recovery algorithm, gathering the votes of 
all self-approved entries. If an entry was chosen in term $t$,
then it is included in the resent self-approved entries because
a classic quorum is required to elect a leader.
The newly elected leader also inserts the entry with the most votes, 
the chosen value, just as the leader in term $t$ did.
Thus, any entry committed by a leader in term $t$ must
also be committed by all leaders in terms $s > t$.
\end{proof}

We now consider membership changes. 
When a site joins or leave the system, a new configuration entry
is created and committed. As with other log entries, not
all sites necessarily have the new configuration entry inserted.
These sites will not update their quorum sizes. If such a site
loses contact with the current leader, it will start an election,
and could wrongfully think its the new leader due to making decisions
on smaller quorum sizes. This can lead to two leaders being
elected in a term, allowing safety to be violated.

As proven in Ongaro's thesis~\cite{ongaro2014consensus},
safety in classic Raft is preserved 
if only one site joins or leaves at a time. A majority of
the sites reach consensus on membership changes one at a time,
never changing quorum sizes in a such a way to allow two leaders
to be elected.
In Fast Raft, when sites announce their joins and leaves, 
only one site can join or leave at a time, as the leader will
only process join or leave requests one at a time. For each
request, consensus on the new configuration must occur, and thus
at least a majority of sites in the system know about the quorum
size changes.
As Fast Raft's leader election only differs in the
recovery algorithm when a leader is elected, safety
is preserved for announced joins and leaves, just as in classic Raft. 

If sites leave the system silently, the true quorum sizes may decrease.
Thus,
consensus or election decisions may be based on quorum sizes 
that are larger than necessary. 
However, using larger quorums cannot lead to two leaders being elected
and violating safety.
It may affect liveness, as discussed in the next subsection.

\subsection{Liveness}

In order to guarantee liveness, we require no concurrent
proposals, otherwise it is possible proposed entries will be overwritten.
Further, we require at least a classic quorum of messages are
delivered to the leader to complete consensus.


It is also required that if a majority of
sites leave the system silently, that there be a leader 
that remains active long enough to detect and commit the configuration change.
If a majority of sites leave, consensus on new log entries 
cannot be completed on the fast or
classic tracks. After the leader has a member timeout with 
the sites, the leader can insert a 
new configuration and decrease the leader's
perception of quorum sizes, and commit a new configuration entry.
However, if the leader fails before committing the new configuration
entry, or was part of the sites that silently left the system,
the remaining sites will be unable to elect a new leader.
In such a scenario, the system is deadlocked.


\section{C-Raft}
\label{craft.sec}


The goal of C-Raft is improve the throughput of consensus
in globally distributed systems where
latency between distant sites can be very large. C-Raft mitigates
this by performing two levels of consensus: local consensus
with nearby sites, and thus have lower message latency, and then global consensus on batches
of locally committed entries. 

\subsection{System Model}

As in Fast Raft, the system has asynchronous communication
between all sites. Sites form a set of clusters.
The membership of each cluster can change over time, and each site can be a 
member of one or more clusters. 
We assume sites are aware of which cluster they are a member of,
but not necessarily the membership of other sites.
Further, the number of clusters may change over time.
We separate the means of communication within clusters, 
\emph{intra-cluster}, and across clusters, \emph{inter-cluster}.
Both are asynchronous with message loss, but we assume that 
intra-cluster communication is ``easier", whether this is due to
communication cost, bandwidth, or minimum message propagation time.

In addition to the global log, each cluster 
also replicates its own \emph{local log}.
The local log serves two purposes: buffering
of log entries for the global log, and state
replication for inter-cluster consensus.
Within each cluster, sites propose entries to
be placed in the local log of that cluster. Sites of the cluster reach
consensus on the entry, at which point the entry is committed to the local log.
Periodically, the leader of a cluster proposes a batch of local entries 
to be committed to the global log. Batches may be created and proposed
based on how many entries have been placed in
the local log, an amount of time passing, or a user-defined need.

During inter-cluster consensus, it is possible for a cluster leader to fail and
a new one be elected. The local log of the cluster
is used here to ensure a leader's state for
the global log is passed on to new leaders if there is a failure during
at any step of inter-cluster consensus. This is achieved through consensus
on special log entries for inter-cluster consensus state replication.

The hierarchical structure allows proposers
to have their entries replicated locally, with the
guarantee that the entries will eventually be replicated to other
clusters and totally-ordered in the global log. 
C-Raft utilizes this model to improve throughput of consensus
at the global scale.

\subsection{Algorithm}
C-Raft is defined with two levels of Fast Raft:
one for intra-cluster consensus, and one for inter-cluster consensus.
The intra-cluster consensus algorithm is simply the execution of Fast Raft over the members of the cluster.  In inter-cluster consensus, the leaders of each cluster act as
members in a configuration. The members of this configuration are
followers of inter-cluster consensus, and \emph{local leaders} for
intra-cluster consensus. 
A \emph{global leader}
is elected by the local leaders, and Fast Raft is executed the same as specified in
Section~\ref{fast.sec} with some small, but
important, changes. Inter-cluster consensus pseudocode is presented
below, and the changes to Fast Raft are marked in blue and by a $\dagger$ symbol. 
Much of Fast Raft remains the same in inter-cluster consensus, and is
left out of the pseudocode here.


Sites now contain state for both cluster levels, intra-cluster state and 
inter-cluster state, with the same types of variables.
Variables associated with intra-cluster consensus are prefixed with
\emph{local}, and they are prefixed with \emph{global} for inter-cluster consensus.
Sites part of intra-cluster consensus propose entries for the
local log, and local leaders propose batches of local $log$ entries
to be inserted in the global $log$.
Note, that as a local leader is a member of the global configuration,
if a new local leader is elected, it must send a join request to enter the global
configuration.  

During inter-cluster consensus, local leaders receive entries
to insert into their global log either from proposers or the global leader.
The important difference from Fast Raft here is that the 
local leader must first run intra-cluster consensus 
before inserting the entry into its global log. The local leader proposes a special log entry for the local log. We call this a
\emph{global state entry}. The purpose of this log entry is to
replicate the local leader's state in inter-cluster consensus.
The global state entry contains in it the global entries that the local leader
has inserted into the global log, and their indices. 
Sites within the cluster run Fast Raft on this special log entry.
Once the entry is committed to the local log, 
the local leader inserts the entry into its global log.
The purpose of the global state entry is to ensure that 
if the local leader of that cluster fails, the global entry that
was inserted is still available for future local leaders. 

It is possible that the local leader overwrites
its global log entry if the global leader sends a different entry at the 
same index.
In this case, the local leader proposes a global state entry to 
its cluster that overrides a previous one. 
Similar to configuration entries, a site determines its global
state based on the global state entry last inserted in the local log.

\vspace{0.5mm}
\noindent\fbox{%
\parbox{\linewidth}{%
\textbf{When follower receives a proposed entry $e$ for index $i$:}
\begin{enumerate}
    \setlength\itemsep{0em}
    \item If entry is duplicate and committed, notify proposer.
    \item If there is no entry at index $i$, insert entry to log at index $i$.
    \item Set $e.insertedBy = self$
    \item \gen{Run intra-cluster consensus on global 
        state entry for $e$ at $k$$^\dagger$.}
    \item \gen{Once entry is committed to the local log$^\dagger$,} 
        send entry and $commitIndex$ to $leaderId$.
\end{enumerate}}}

\vspace{0.5mm}
\noindent\fbox{%
\parbox{\linewidth}{%
\textbf{Periodically run by the leader:}
\begin{enumerate}
    \setlength\itemsep{0em}
    \item While there exists a $k = commitIndex+1$ for which 
        at least a classic quorum of votes has been received \gen{for the global log$^\dagger$}:
    \item If there is an entry $e$ with $e.insertedBy = leader$, then
        insert to log.
    \item Else, insert entry $e$ from $possibleEntries[k]$ with highest number 
        of votes. Break ties arbitrarily.
    \item \gen{Run intra-cluster consensus on global 
        state entry for $e$ at $k$$^\dagger$.}
    \item \gen{Once entry is committed to the local log$^\dagger$,} 
        update $fastMatchIndex[i]$ for all $i$ that voted for entry $e$.
\end{enumerate}
$\cdots$}}

\noindent\fbox{%
\parbox{\linewidth}{%
\textbf{When a follower receives AppendEntries message:}
\\$\cdots$
\begin{enumerate}
    \setlength\itemsep{0em}
    \setcounter{enumi}{3}
    \item If an existing entry conflicts with a new one, 
        overwrite the existing entry.
    \item Set all new entries $e.insertedBy = leader$
    \item \gen{Run intra-cluster consensus on global 
        state entry for $e$ at $k$$^\dagger$.}
\end{enumerate}}}
\noindent\fbox{%
\parbox{\linewidth}{%
\begin{enumerate}
    \setlength\itemsep{0em}
    \setcounter{enumi}{5}
    \item \gen{For both local and global variables$^\dagger$:}
            if $leaderCommit > commitIndex$, set $commitIndex$ to
        the minimum of $leaderCommit$ and $lastLogIndex$

\end{enumerate}}}
\vspace{0.01mm}

In inter-cluster consensus, a global state entry is run 
through intra-cluster consensus
when the local leader receives a proposal for a global entry,
when the global leader chooses an entry to be placed in the global log,
and when a local leader receives a global $AppendEntries$ message.

Sites now contain a global $commitIndex$ indicating what entries are
committed in the global log.
The global leader lets local leaders know when 
entries are committed through heartbeat messages, just as in Fast Raft. 
Local leaders now need to include their 
global $commitIndex$ in the $AppendEntries$ message to let followers at
the local level know which global entries are committed.

\subsection{Cluster Membership}

Similar to how Fast Raft's configuration defines the
membership of sites in the cluster,
C-Raft's local leaders are defined by a global configuration.
A new cluster can be formed if a new local leader is added 
to the global configuration. 

\vspace{0.5mm}
\noindent\fbox{%
\parbox{\linewidth}{%
\textbf{When a site wants to \gen{make a new cluster$^\dagger$}:}
\begin{enumerate}
    \setlength\itemsep{0em}
    \item Send a join request to \gen{local leaders$^\dagger$}.
    \item If join not accepted in join timeout, resend request.
\end{enumerate}}}

\noindent\fbox{%
\parbox{\linewidth}{%
\textbf{Upon receiving a join request:}
\begin{enumerate}
    \setlength\itemsep{0em}
    \item If not the leader, redirect to the leader.
    \item If a duplicate request, ignore.
    \item Catch up joining site.
\end{enumerate}}}

\noindent\fbox{%
\parbox{\linewidth}{%
\textbf{Upon catching up site:}
\begin{enumerate}
    \item Create a \gen{global$^\dagger$} log entry for the new configuration
        including the new site.    
    \item Start consensus at global level on new log entry. 
    \item When consensus reached (the entry is committed), 
        notify the new site.
\end{enumerate}}}
\vspace{0.01mm}

To form a new cluster, a site 
sends a global join request a member of the global configuration. 
The global leader catches up the site on global log entries if
necessary, then insert the new configuration to the log. It
then proposes the new global configuration entry for global
consensus.
If the configuration change is committed, the requesting site 
becomes the local leader of the new cluster. 
Sites that enter this new cluster may now send join requests 
to the local leader to start proposing entries for intra-cluster
consensus. 


\subsection{Safety}\label{safety2.sec}

We prove that C-Raft satisfies safety.

\begin{theorem}
    C-Raft satisfies safety in Definition~\ref{safety.def}.
\end{theorem}

\begin{proof}
Let us consider the global level of C-Raft. This is defined the same
as Fast Raft, with the exception of the recursive call of Fast Raft 
at the local level. Before an entry is inserted into the global log of any
site, Fast Raft is run on global state entry at the local level.
As proven in Section~\ref{fast.sec}, Fast Raft
satisfies safety, thus we can be certain that the entry will be 
committed to the local log of any future local leaders. 
Once an entry is inserted into the local
leader's global log, the cluster maintains the state of the local
leader in inter-cluster consensus, and the global level
of Fast Raft acts no differently than its local version. 
Fast Raft is proven to satisfy safety, and thus C-Raft satisfies safety.
\end{proof}

\subsection{Liveness}
As the global level of C-Raft is defined with a modified Fast Raft, 
similar liveness conditions apply.
Instead of liveness conditions being based on sites failing, 
the conditions require clusters to not fail.
If the conditions for liveness in Fast Raft do not hold within a cluster,
we consider this to mean the cluster has failed. For example, if 
a majority of the cluster has failed, then the local leader cannot 
insert global log entries and will block on consensus step for the
global state entry.
For liveness to be guaranteed at the global level, liveness must be
guaranteed for intra-cluster consensus in enough clusters 
for inter-cluster consensus to continue.

\section{Experiments}
\label{exp.sec}

To showcase the performance of Fast Raft and C-Raft, we performed
experiments on Amazon Web Services (AWS). Each site 
was set up as a separate EC2 AWS instance running Ubuntu 18.04 LTS. 
To simulate clusters, instances
were started in different regions around the world including
North America, South America, Europe, and Asia. Sites assigned to
the same AWS region were considered part of the same cluster. 
Roundtrip message latency varies between $10$ to $300$~ms between
AWS regions and is less than $1$~ms within regions.
To force certain percentages of message loss, we changed
traffic control settings in Linux using the $tc$ command.

We implemented classic Raft, Fast Raft, and C-Raft in Python 3.6 using UDP
sockets for communication. 
The leader's heartbeat timer was set to $100$ ms for intra-cluster consensus,
and $500$ ms for inter-cluster consensus. 
To measure latency of committing
an entry, the proposer started a timer when first proposing an entry and
stopped the timer when the leader notified it that the entry was committed.
The proposer only proposed a new entry after the previous
entry was committed.

\subsection{Classic Raft vs. Fast Raft}

First, we compared the commit latency of classic Raft and Fast Raft in a single
cluster. We chose a site at random
to be the proposer and measured the average latency for
entries committed over $100$ trials when using classic Raft and Fast Raft.  
In the experiments, we had five sites in one region and
varied message loss between $0\%$ and $10\%$.

\begin{figure}[t]
    \centering
    \includegraphics[scale=0.25]{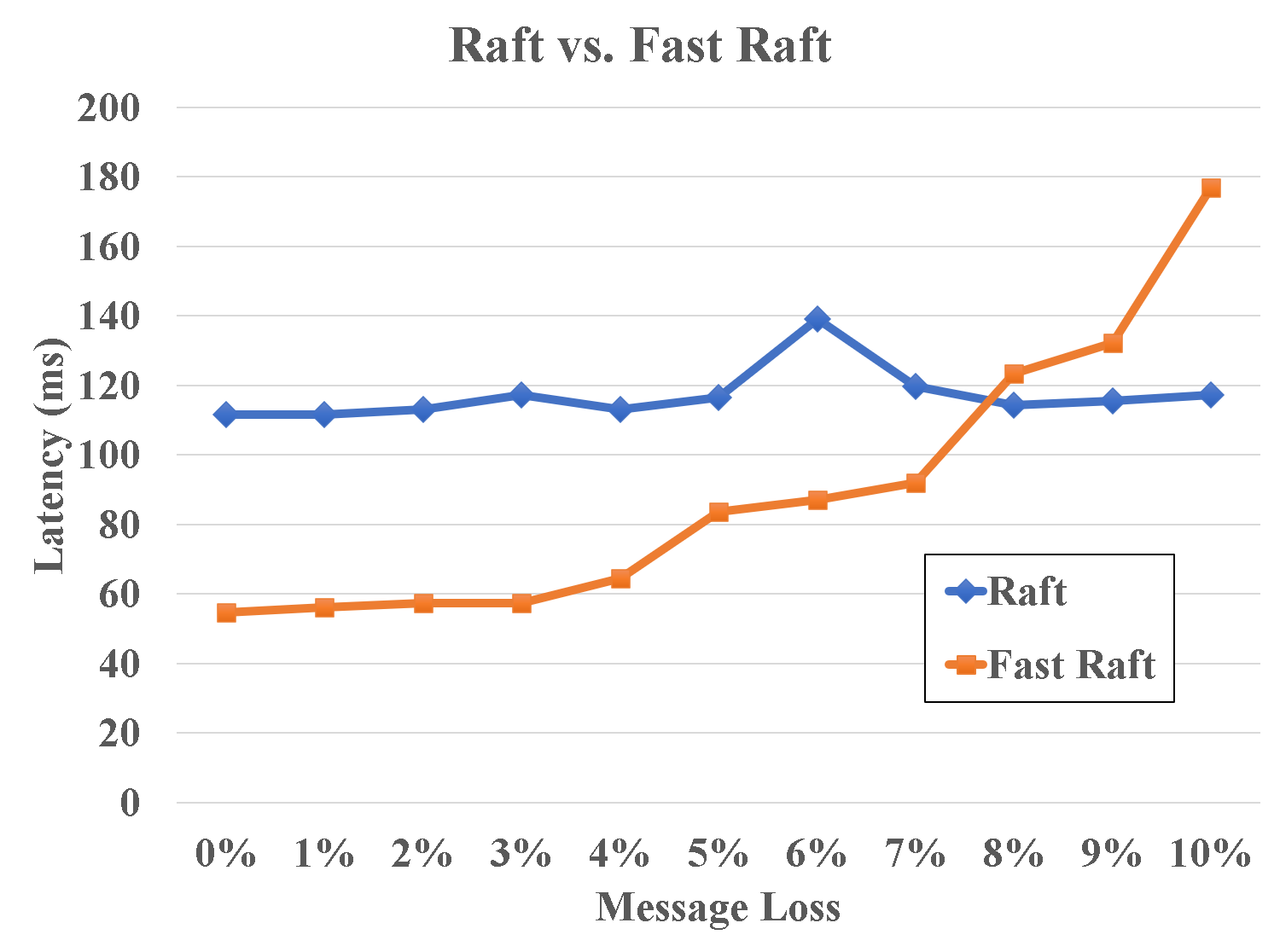}
    \caption{The average latency of committing entries
    in classic Raft and Fast Raft, with five sites and different percentages of message loss. 
    }
    \label{raft_v_fraft.fig}
\end{figure}

Figure~\ref{raft_v_fraft.fig} shows the results of the experiment. 
When message loss was low, Fast Raft achieved about half the latency
as classic Raft. 
However, as message loss increased, Fast Raft started to
degrade in performance while classic Raft maintained similar latency.
As more messages were dropped, the classic-track was used more in Fast
Raft, causing it to face an extra message round more often.
This reinforces the observation that Fast Raft
is best used when message loss is not common.

\subsection{Silently Leaving a Cluster}

Next, we studied the effect of sites silently 
leaving a cluster on commit latency in Fast Raft. We 
started with a cluster of five sites, had two of the
sites silently leave the cluster, and measured the
latency on committing entries during this period. 
The fast quorum size before leaves was four. As such, during the 
period before the leader detected the leaves, proposed entries
used the classic track. Message loss was set to
$5\%$ and the leader's member timeout occurred after five missed
heartbeat responses. 
\begin{figure}[t]
    \centering
    \includegraphics[scale=0.4]{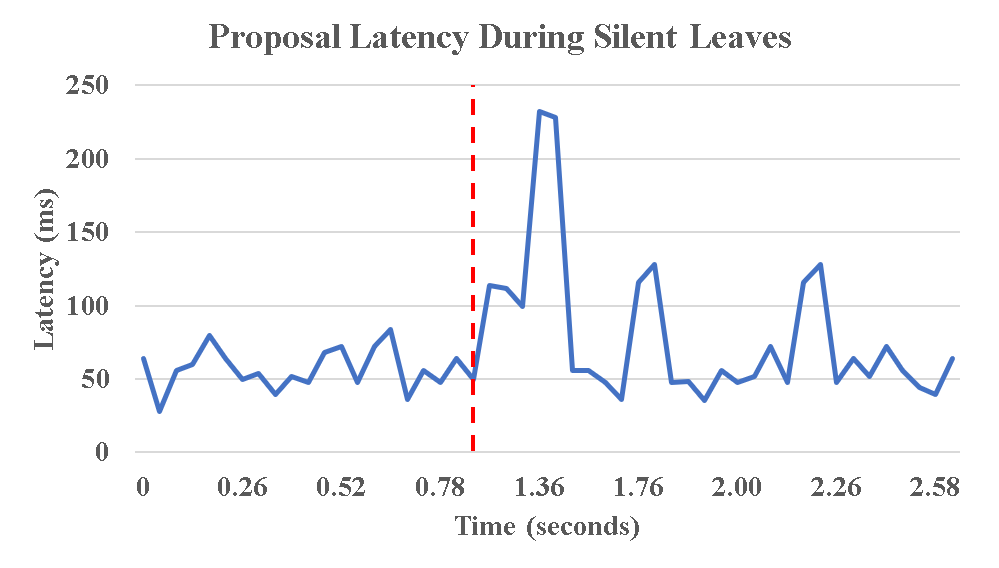}
    \caption{Latency on proposed entries being committed in 
    Fast Raft in a cluster of five sites. The vertical red line 
    indicates when two sites silently leave the cluster.}
    \label{exp3.fig}
\end{figure}

Figure~\ref{exp3.fig} shows the results of this experiment.
The vertical red line indicates when the sites left the cluster.
Before this point, the proposer was typically able to use the fast
track. The variability in latency of proposals here can be due to
proposals occurring closer or further from the heartbeat timeout
of the leader, or due to message loss causing the leader to use
the classic track. Once the leave occurred, there was a brief period 
where the fast track was no longer available to the proposer. The
very large spike to above $200$ ms in this section was likely due to the concurrent
proposals with the leader for a configuration change. 
After this point, proposal latency returned to a range of $50$ to $100$ ms.

\subsection{Classic Raft vs. C-Raft}
Finally, we compared the 
throughput of Raft and C-Raft. Sites in
the same AWS region formed a cluster. We
chose one proposer at random per cluster. Each proposer waited
until its last proposed entry is committed before proposing
another. We compared throughput based on 
how many entries were committed to the global log in classic Raft 
and C-Raft, averaged over five $3$-minute trials. 

\addtolength{\textheight}{0cm}
For the C-Raft
implementation, a cluster proposed a batch of entries
to the global log after ten entries were committed in the local
log. We tested with $20$ sites total, split evenly over
a varying number of clusters. 
Note, there were more proposers in the system as the number
of clusters increased for both C-Raft and Raft. 
\begin{figure}[t]
    \centering
    \includegraphics[scale=0.25]{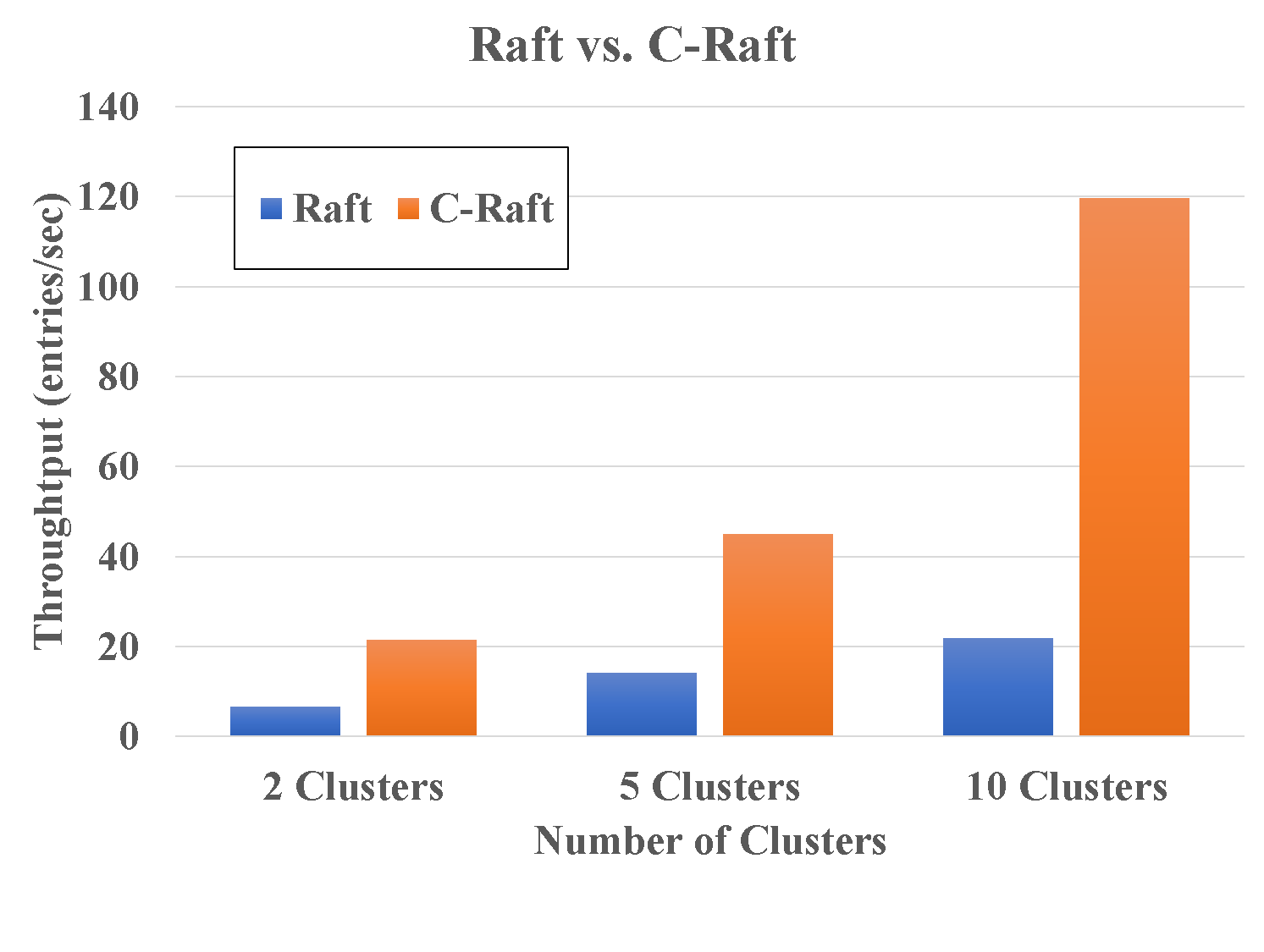}
    \caption{Average throughput of committing entries in classic Raft versus C-Raft. Experiment run for
    with $20$ sites total. Each cluster is in a different AWS region.
    }
    \label{raft_v_craft.fig}
\end{figure}
The results are shown in Figure~\ref{raft_v_craft.fig}.
C-Raft showed significant improvements over classic Raft in these experiments,
with a $5x$ throughput increase for $10$ clusters.
In addition to using Fast Raft for intra-cluster and inter-cluster consensus, 
C-Raft was able to run consensus quickly within clusters where latency was low, then
only run slower consensus at the global level between leaders (rather
than all sites) periodically. Batching
entries for consensus has a large impact when sites are distributed
geographically in this manner. 

\section{Related Work}
\label{related.sec}

As discussed before, the Paxos algorithm~\cite{lamport2001paxos} 
is a popular consensus algorithm.
Fast Paxos was presented in~\cite{lamport2006fast} with the goal
of removing a round of messaging from the classic Paxos in
the absence of concurrent proposals.
A common critique
of these algorithms is how much of the implementation is left 
up to the reader. Some users looked for consensus protocols that were
more well-defined, and more ready for implementation.
Raft~\cite{ongaro2014consensus} is a consensus protocol that filled this need. 
Fast Raft takes the same approach to speeding up Raft as
Fast Paxos does for Paxos, while maintaining
the understandability design philosophy of Raft.

Many Paxos variations have been proposed to improve upon
the original algorithm. 
Generalized Paxos~\cite{lamport2005generalized} aims to 
mitigate the shortcomings of Fast Paxos by defining a 
partial ordering of committed entries. This change allows
for concurrent proposals of non-conflicting entries to
occur without slowdowns. 
Egalitarian Paxos~\cite{moraru2013there} notes that 
the leader creates a potential bottleneck in the system, and
proposes an algorithm that does not rely on a fixed stable
leader.
Flexible Paxos~\cite{howard2017flexible} makes a simple but strong observation that 
leader election and replication phases can have
different quorum sizes, allowing for smaller replication
quorum sizes without violating safety.
Dynamic Paxos~\cite{nawab2018dpaxos} builds on Flexible Paxos, applies 
it to a geographically distributed setting, and further
reduces the leader election quorum size.
Applying the improvements of these algorithms to Raft would
be an interesting future direction.

Delegator Paxos~\cite{liu2016d} is a hierarchical form of 
Paxos. Similar to C-Raft, in Delegator Paxos, Paxos is run in separate clusters
for local replication. 
Consensus inside a cluster runs for a specified 
$k$ user requests. After this point, a global leader
is chosen, 
and leaders of each cluster take part in
Paxos at a global level. Batches of entries are replicated to 
other clusters. The major difference from classic Paxos is
that before a follower can accept a global leader's proposed batch, 
it must replicate to a majority inside its cluster to ensure safety
in the case of failures.
C-Raft takes inspiration from this model to define a
hierarchical model with additional benefits of dynamic
clusters and faster consensus.

Another relevant Paxos variant is 
Institutionalised Paxos~\cite{sanderson2012institutionalised}. 
This work also considers a system that is divided into clusters.
To deal with dynamic cluster membership,
the leader keeps track of the number of sites in the cluster
and how the quorum size grows/shrinks.
It is assumed that the leader does not fail or leave the cluster
without handing off the value for the number of sites.
This is impractical in many systems where the leader may
fail or leave the cluster unexpectedly.
Raft and C-Raft, in contrast, deal with membership changes by 
keeping track of the number of sites through reaching consensus on
a configuration.  

\section{Conclusion}
\label{conclusion.sec}

We presented two new consensus algorithms:
Fast Raft, a variation on the Raft consensus algorithm that 
speeds up consensus in typical operation, and C-Raft, which defines 
a hierarchical model of Fast Raft consensus. 
Both algorithms deal with  membership
changes in dynamic networks. We proved safety for both algorithms
and discussed their liveness requirements. Finally, we presented an experimental evaluation of both algorithms in AWS.
Our experiments show that Fast Raft can achieve half the latency of classic Raft 
when message loss is low, and C-Raft can achieve a $5x$ throughput improvement
over classic Raft in a globally distributed
scenario.
In future work, we plan to explore extending C-Raft to support 
partially-ordered log entries, similar to Generalized Paxos~\cite{lamport2005generalized}.



\bibliographystyle{IEEEtran}
\bibliography{references}

\begin{thebibliography}{10}
\providecommand{\url}[1]{#1}
\csname url@samestyle\endcsname
\providecommand{\newblock}{\relax}
\providecommand{\bibinfo}[2]{#2}
\providecommand{\BIBentrySTDinterwordspacing}{\spaceskip=0pt\relax}
\providecommand{\BIBentryALTinterwordstretchfactor}{4}
\providecommand{\BIBentryALTinterwordspacing}{\spaceskip=\fontdimen2\font plus
\BIBentryALTinterwordstretchfactor\fontdimen3\font minus
  \fontdimen4\font\relax}
\providecommand{\BIBforeignlanguage}[2]{{%
\expandafter\ifx\csname l@#1\endcsname\relax
\typeout{** WARNING: IEEEtran.bst: No hyphenation pattern has been}%
\typeout{** loaded for the language `#1'. Using the pattern for}%
\typeout{** the default language instead.}%
\else
\language=\csname l@#1\endcsname
\fi
#2}}
\providecommand{\BIBdecl}{\relax}
\BIBdecl

\bibitem{chubby}
M.~Burrows, ``The chubby lock service for loosely-coupled distributed
  systems,'' in \emph{7th {USENIX} Symp. on Operating Systems Design and
  Implementation}, 2006, pp. 335--350.

\bibitem{isard2007autopilot}
M.~Isard, ``Autopilot: Automatic data center management,'' \emph{Operating
  Systems Review}, vol.~41, pp. 60--67, April 2007.

\bibitem{buterin2013ethereum}
V.~Buterin \emph{et~al.}, ``Ethereum white paper,'' \emph{GitHub repository},
  pp. 22--23, 2013.

\bibitem{rad2009survey}
M.~Rad, A.~Badashian, G.~Meydanipour, M.~Delcheh, M.~Alipour, and H.~Afzali,
  ``A survey of cloud platforms and their future,'' in \emph{Intl. Conf. on
  Computational Science and its Applications}, 2009, pp. 788--796.

\bibitem{tzanakaki2013virtualization}
A.~Tzanakaki, M.~Anastasopoulos, G.~Zervas, B.~Rofoee, R.~Nejabati, and
  D.~Simeonidou, ``Virtualization of heterogeneous wireless-optical network and
  it infrastructures in support of cloud and mobile cloud services,''
  \emph{IEEE Communications Magazine}, vol.~51, no.~8, pp. 155--161, 2013.

\bibitem{mqttrss}
\BIBentryALTinterwordspacing
``Mqtt.'' [Online]. Available: \url{http://mqtt.org/ (visited on 04/13/2020)}
\BIBentrySTDinterwordspacing

\bibitem{lamport2001paxos}
L.~Lamport, ``Paxos made simple,'' \emph{ACM Sigact News}, vol.~32, no.~4, pp.
  18--25, 2001.

\bibitem{ongaro2014consensus}
D.~Ongaro, ``Consensus: Bridging theory and practice,'' Ph.D. dissertation,
  Stanford University, 2014.

\bibitem{lamport2006fast}
L.~Lamport, ``Fast paxos,'' \emph{Distributed Computing}, vol.~19, no.~2, pp.
  79--103, 2006.

\bibitem{lamport2009vertical}
L.~Lamport, D.~Malkhi, and L.~Zhou, ``Vertical paxos and primary-backup
  replication,'' in \emph{Proc. 28th {ACM} Symp. on Principles of Distributed
  Computing}, 2009, pp. 312--313.

\bibitem{nawab2018dpaxos}
F.~Nawab, D.~Agrawal, and A.~El~Abbadi, ``D-paxos: Managing data closer to
  users for low-latency and mobile applications,'' in \emph{Int. Conf. on
  Managment of Data}, 2018, pp. 1221--1236.

\bibitem{moraru2013there}
I.~Moraru, D.~G. Andersen, and M.~Kaminsky, ``There is more consensus in
  egalitarian parliaments,'' in \emph{Proc. 24th ACM Symp. on Operating Systems
  Principles}, 2013, pp. 358--372.

\bibitem{liu2016d}
F.~Liu and Y.~Yang, ``D-paxos: building hierarchical replicated state machine
  for cloud environments,'' \emph{IEICE Trans. on Information and Systems},
  vol.~99, no.~6, pp. 1485--1501, 2016.

\bibitem{connolly2005practical}
T.~M. Connolly and C.~E. Begg, \emph{Database systems: a practical approach to
  design, implementation, and management}.\hskip 1em plus 0.5em minus
  0.4em\relax Pearson Education, 2005.

\bibitem{bandyopadhyay2003energy}
S.~Bandyopadhyay and E.~J. Coyle, ``An energy efficient hierarchical clustering
  algorithm for wireless sensor networks,'' in \emph{IEEE INFOCOM 2003},
  vol.~3, 2003, pp. 1713--1723.

\bibitem{skala2015scalable}
K.~Skala, D.~Davidovic, E.~Afgan, I.~Sovic, and Z.~Sojat, ``Scalable
  distributed computing hierarchy: Cloud, fog and dew computing,'' \emph{Open
  Journal of Cloud Computing}, vol.~2, no.~1, pp. 16--24, 2015.

\bibitem{shokri2015privacy}
R.~Shokri and V.~Shmatikov, ``Privacy-preserving deep learning,'' in
  \emph{Proc. 22nd ACM SIGSAC Conference on Computer and Communications
  Security}, 2015, pp. 1310--1321.

\bibitem{dasu2018unchain}
T.~Dasu, Y.~Kanza, and D.~Srivastava, ``Unchain your blockchain,'' in
  \emph{Proc. Symp. on Foundations and Applications of Blockchain}, 2018, pp.
  16--23.

\bibitem{fischer1985impossibility}
M.~J. Fischer, N.~A. Lynch, and M.~S. Paterson, ``Impossibility of distributed
  consensus with one faulty process,'' \emph{Journal of the ACM}, vol.~32,
  no.~2, pp. 374--382, 1985.

\bibitem{zhao2015fast}
W.~Zhao, ``Fast paxos made easy: Theory and implementation,''
  \emph{International Journal of Distributed Systems and Technologies}, vol.~6,
  no.~1, pp. 15--33, 2015.

\bibitem{lamport2005generalized}
L.~Lamport, ``Generalized consensus and paxos,'' \emph{Technical Report
  MSR-TR-2005-33, Microsoft Research}, 2005.

\bibitem{howard2017flexible}
H.~Howard, D.~Malkhi, and A.~Spiegelman, ``Flexible paxos: Quorum intersection
  revisited,'' in \emph{20th Intl. Conf. on Principles of Distributed Systems},
  2017.

\bibitem{sanderson2012institutionalised}
D.~Sanderson and J.~Pitt, ``Institutionalised paxos consensus.'' in \emph{20th
  European Conf. on Artificial Intelligence}, 2012, pp. 714--719.

\end{thebibliography}

\end{document}